\documentclass[a4paper, 11pt]{article}

\usepackage{epsfig}
\usepackage{amsmath}
\usepackage{amssymb}
\usepackage{latexsym}
\usepackage{amsfonts}
\usepackage{amsthm}
\usepackage{color}
\usepackage[numbers,sort&compress]{natbib}
\usepackage{graphicx}
\ifx\pdfoutput\undefined  \else
 \fi
\usepackage[centerlast]{caption2}

\newtheorem{lemma}{Lemma}
\newtheorem{observation}{Observation}

\newtheorem{theorem}{Theorem}
\numberwithin{figure}{section} \numberwithin{definition}{section}
\numberwithin{observation}{section} \numberwithin{lemma}{section}
\numberwithin{theorem}{section} \numberwithin{corollary}{section}

\begin{document}

\title{{\bf The crossing number of the generalized Petersen graph $\bf P(10,3)$ is six} \footnote{The
research is supported by NSFC (60973014, 61170303)}}
\author{
\ Yuansheng Yang\footnote {corresponding author's
email : yangys@dlut.edu.cn}, \ Baigong Zheng, \ Xirong Xu \\
School of Computer Science and Technology, \\
Dalian University of Technology, Dalian, 116024, P. R. China }
\date{}
\maketitle
\begin{abstract}
The crossing number of a graph is the least number of crossings of
edges among all drawings of the graph in the plane. In this article,
we prove that the crossing number of the generalized Petersen graph
$P(10,3)$ is equal to 6.
\end{abstract}
\section{Introduction}

Let $n\geq3$ and $k\in \mathbb{Z}_n\backslash \{0\}$. The
$generalized\mbox{ } Petersen\mbox{ } graph\mbox{ } P(n,k)$ is
defined on the set of vertices $\{x_i,y_i|i\in \mathbb{Z}_n\}$ with
edges $x_ix_{i+1}$, $x_iy_i$ and $y_iy_{i+k}$ \cite{S97}. The {\it
crossing number} of $G$, denoted by $cr(G)$, is the smallest number
of pairwise crossings of edges among all drawings of $G$ in the
plane. Fiorini showed the values of $cr(P(n,k))$ for $n$ up to 14 in
\cite{F86} and the values are extended for $n$ up to 16 in
\cite{LY09}. And in \cite{F86} the smallest unresolved cases result
to be the graphs $P(10,4)$ and $P(10,3)$. For the first one, it was
resolved by Sara\u{z}in in \cite{S97}. For the second one, McQuillan
and Richter in \cite{MR92} showed that $cr(P(10,3))\geq5$ and with
the result verified by computer that $cr(P(10,3))=6$, Richter and
Salazar proved $cr(3k+1,3)=k+3$ in \cite{RS02}. Fiorini and Gauci
\cite{FG03} extended their result by showing that $P[3k,k]$ also has
crossing number $k$ for all $k\geq4$.

The main result in this article is the following theorem
\begin{theorem}\label{Theorem: cr(P(103))= 6}
$cr(P(10,3))=6$.
\end{theorem}

The rest of this paper is organized as follows. In Section 2 we
introduce some technical notations and tools, while in Section 3 we
prove Theorem \ref{Theorem: cr(P(103))= 6}.
\section{Preliminaries}

\indent \indent Let $G$ be a simple connected graph with vertex set
$V(G)$ and edge set $E(G)$. For $S\subseteq V(G)$, let $[S]$ be the
subgraph of $G$ induced by $S$. Let $P_{v_1\cdots v_n}$ be the
$path$ with $n$ vertices $v_1,\cdots ,v_n$ and let $C_{v_1\cdots
v_n}$ be the $cycle$ with $n$ vertices $v_1,\cdots ,v_n$. We denote
by $v_iv_j$ the edge with ends in the vertices $v_i$ and $v_j$. The
neighborhood of a vertex $v$ is denoted by $N(v)$, formed by all
vertices adjacent to $v$ in graph. And the neighborhood of $S$,
denoted by $N(S)$, is the union of $N(v)$ for all $v\in S$.
Additionally, the degree of a vertex $v$ in a graph $G$ is denoted
by $d_G(v)$.

A drawing of $G$ is said to be a $good$ drawing, provided that no
edge crosses itself, no adjacent edges cross each other, no two
edges cross more than once, and no three edges cross in a point. It
is well known that the crossing number of a graph is attained only
in good drawings of the graph. So we always assume that all drawings
throughout this paper are good drawings. For a good drawing $D$ of a
graph $G$, let $\nu(D)$ be the number of crossings in $D$. In a
drawing $D$, if an edge is not crossed by any other edge, we say it
is $clean$ in $D$.

Let $A$ and $B$ be two disjoint subsets of an edge set $E$. In a
drawing $D$, the number of the crossings formed by an edge in $A$
and another edge in $B$ is denoted by $\nu_D(A,B)$. The number of
the crossings that involve a pair of edges in $A$ is denoted by
$\nu_D(A)$. Then $\nu_D(A\cup B)=\nu_D(A)+\nu_D(B)+\nu_D(A,B)$ and
$\nu(D)=\nu_D(E)$.

In this paper, we will use the term $``face"$ in planar drawings,
and $``region"$ in nonplanar drawings. By a $line\mbox{ } segment$,
we mean a curve incident with vertices or crossings. The $bound$ of
a region $R$ is the boundary of the open set $R$ in the usual
topological sense. Two drawings of $G$ are $isomorphic$ if and only
if there is an incidence preserving one-to-one correspondence
between their vertices, edges, parts of edges and regions. We denote
$G\cong H$ if graph $G$ is isomorphic to graph $H$ and $G\not\cong
H$ if not.

\begin{figure}[ht]
\centering\includegraphics[scale=1.0]{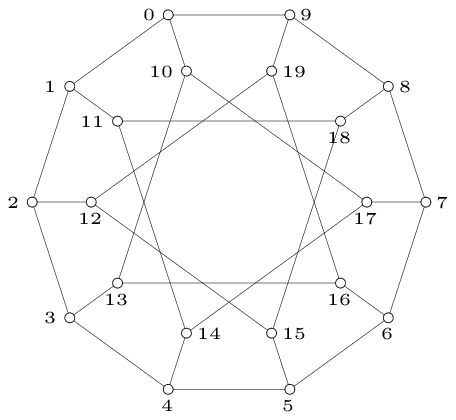}\hspace{20bp}
\centering\includegraphics[scale=1.0]{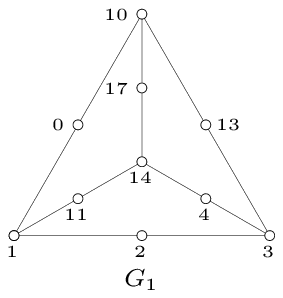}\hspace{20bp}
\centering\includegraphics[scale=1.0]{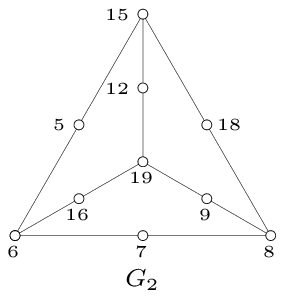}
\caption{\small{$P(10,3)$} and its two subgraphs}\label{fig:(P103)}
\end{figure}

Now we will introduce the graph $P(10,3)$. Let
$$\begin{array}{rlll}
V_{1}&=\{v_{0},v_{1},v_{2},v_{3},v_{4},v_{10},v_{11},v_{17},v_{13},v_{14}\},\\
V_{2}&=\{v_{5},v_{6},v_{7},v_{8},v_{9},v_{15},v_{16},v_{12},v_{18},v_{19}\}.
\end{array}$$
For $i=1,2$, let $G_{i}=[V_{i}]$, $E_i=E(G_{i})$. Let
$E_{12}=\{uv:u\in V_1\wedge v\in V_2\}$ (see Figure
\ref{fig:(P103)}).

For $i=1,2$, let
$$\begin{array}{rlll}
V^{2}_{i}&=\{v:v\in V_{i}\wedge d_{G_i}(v)=2\},\\
V^{3}_{i}&=\{v:v\in V_{i}\wedge d_{G_i}(v)=3\}.\\
\end{array}$$
Then $|V^{2}_{i}|=6$, $|V^{3}_{i}|=4$ for $i=1,2$.

To prove Theorem \ref{Theorem: cr(P(103))= 6}, we need to analyze
certain drawings of these two subgraphs. Therefore, the following
observation and lemmas related to their crossings in different
situations will be useful for our later proofs.

\begin{observation}\label{observation: boundary to 3-degree}
For $i=1,2$, and any $v\in V^{3}_i$, there exists just one $u\in
V^{3}_{3-i}$ such that $N(N(v))\cap V_{3-i}=V^2_{3-i}-N(u)$.
\end{observation}

By this observation, we could have the following straightforward
lemma.

\begin{lemma}\label{lemma: boundary to 3-degree cr>=1}
For $i=1,2$, and any $v\in V^{3}_{i}$, if $\nu_D(E_{3-i})=0$ and all
vertices of $V^2_{i}-N(v)$ lie in the same region of $G_{3-i}$, then
$\nu_D(E_{3-i},E_{12})\geq1$.
\end{lemma}

For any region $R$ of $G_{3-i}$, we define $V_{in}(R;G_{i})=\{v:
v\in V^{2}_{i}$ and $v \mbox{ lies in }R \}$ and
$V_{out}(R;G_{i})=\{v: v\in V^{2}_{i}$ and $v \mbox{ lies in the
outside of }R\}$. Then $|V_{in}(R;G_{i})|\leq 6$,
$|V_{out}(R;G_{i})|\leq 6$, and we can have the following lemmas.

\begin{lemma}\label{lemma: two 2-degree >=3}
If $|V_{in}(R;G_i)|\geq2$ and $|V_{out}(R;G_i)|\geq2$, then the
edges of $E_i$ cross the bound of $R$ at least three times,
$\nu_D(E_1,E_2)\geq3$.
\end{lemma}

\begin{proof}
By contradiction. Suppose $\nu_D(E_1,E_2)\leq 2$. Without loss of
generality, we may assume $|V_{in}(R;G_i)|\leq|V_{out}(R;G_i)|$.
Then $|V_{in}(R;G_i)|\leq 1$ (see Figure \ref{fig:cr(E1E2)<=2}), it
contradicts to $|V_{in}(R;G_i)|\geq2$. So Lemma \ref{lemma: two
2-degree
>=3} holds.
\begin{figure}[ht]
\centering\includegraphics[scale=1]{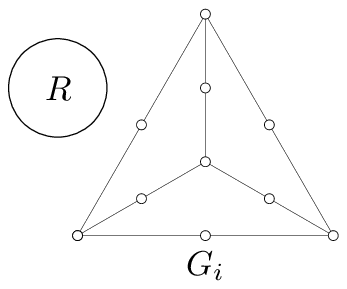}
\centering\includegraphics[scale=1]{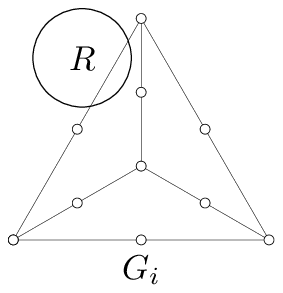}
\centering\includegraphics[scale=1]{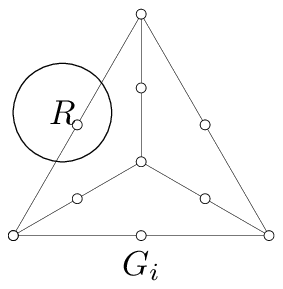}
\caption{\small{All situations for
$\nu_D(E_1,E_2)\leq2$}}\label{fig:cr(E1E2)<=2}
\end{figure}
\end{proof}

\begin{lemma}\label{lemma: cr(G12) >= 3}
For any region $R$ of $G_{i}$, $\nu_D(E_i)+\nu_D(E_i,E_{12})\geq
|V_{in}(R;G_{3-i})|-3$ for $i=1,2$.
\end{lemma}

\begin{proof}Since $|V_{in}(R;G_{3-i})|\leq6$, the conclusion is
straightforward for $\nu_D(E_i)\geq 3$. In Figure \ref{fig:v(Ei)<3},
we enumerate all the different situations in which there should be
the most vertices of $V^{2}_i$ on the boundary of $R$ for
$\nu_D(E_i)\leq 2$. There are at most $3+\nu_D(E_i)$ vertices of
$V^{2}_i$ on the boundary of $R$. So we have $\nu_D(E_i,E_{12})\geq
|V_{in}(R;G_{3-i})|-(3+\nu_D(E_i))$, that is
$\nu_D(E_i)+\nu_D(E_i,E_{12})\geq |V_{in}(R;G_{3-i})|-3$.

\begin{figure}[ht]
\centering\includegraphics[scale=0.8]{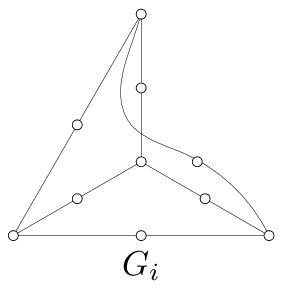}
\centering\includegraphics[scale=0.8]{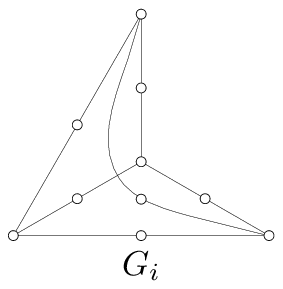}
\centering\includegraphics[scale=0.8]{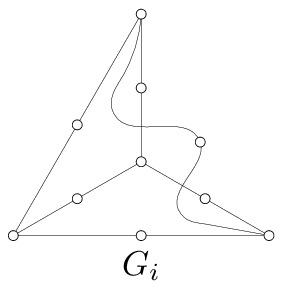}
\centering\includegraphics[scale=0.8]{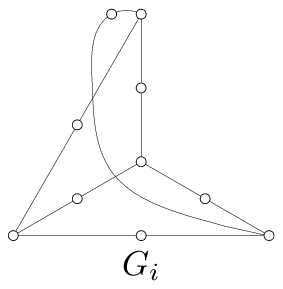}
\centering\includegraphics[scale=0.8]{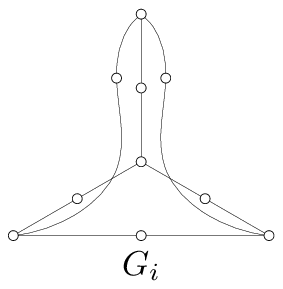}
\centering\includegraphics[scale=0.8]{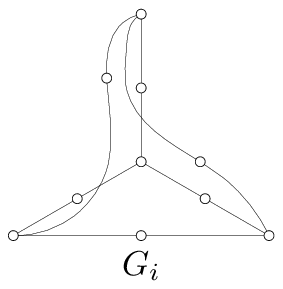}
\caption{\small{All different situations for
$\nu_D(E_i)<3$}}\label{fig:v(Ei)<3}
\end{figure}
\end{proof}

\begin{lemma}\label{lemma:  cr(E1,E2)=0}
If $\nu_D(E_1,E_2)=0$, then
$\nu(D)\geq\nu_D(E_1)+\nu_D(E_1,E_{12})+\nu_D(E_2)+\nu_D(E_2,E_{12})\geq6$.
\end{lemma}

\begin{proof}
Since $\nu_D(E_1,E_2)=0$, for $i=1,2$, $G_{3-i}$ lies in one region
of $G_i$, say region $R_{i}$, and $|V_{in}(R_i;G_{3-i})|=6$. By
Lemma \ref{lemma: cr(G12)
>= 3}, we have $\nu_D(E_i)+\nu_D(E_i,E_{12})\geq
|V_{in}(R_{i};G_{3-i})|-3=6-3=3$. Hence,
$\nu(D)\geq\nu_D(E_1)+\nu_D(E_1,E_{12})+\nu_D(E_2)+\nu_D(E_2,E_{12})\geq
3+3=6$.
\end{proof}

Now we introduce another graph $H$, which is a special dual graph of
$P(10,3)$. It will be very helpful in our later proof, since we will
enumerate the drawings of $P(10,3)$ according to the possibilities
of $H$.

Let $D$ be a drawing of $P(10,3)$. For $i=1,2$, let $D_i$ be the
subdrawing of $D$ corresponding to the edges of $E_i$, and let
$H=(V_{H},E_{H})$ be a graph corresponding to the drawing $D$, where
$V_{H}=\{u_j: R_j$ is a region of $D_1$ and there exists at lest one
segment of $G_{2}$ which lies in region $R_j$\}, $E_{H}=\{(u_j,u_k):
R_j$ and $R_k$ are two adjacent regions of $D_1$ and there exists at
lest one edge of $E_{2}$ which crosses the common boundary of $R_j$
and $R_k$\} (see Figure \ref{fig: H}). Then $\nu_D(E_1,E_2)\geq
|E_H|$. Furthermore, let $f_n=|\{u: u\in V_H$ and $ d_H(u)=1\}|$,
and we have the following lemma.
\begin{figure}[ht]
\centering\includegraphics[scale=1.0]{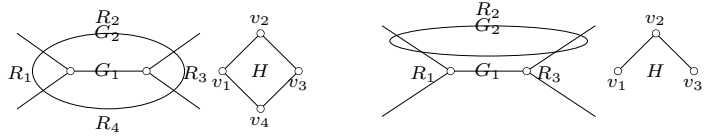}
\caption{\small{Some examples of $H$}}\label{fig: H}
\end{figure}

\begin{lemma}\label{lemma: cr(E_1,E_2)>=E_H}
$\nu_D(E_1,E_2)\geq 2(|V_H|-1)$ for $H\cong P_{|V_H|}$ and
$\nu_D(E_1,E_2)\geq |E_H|+f_n$ for $H\not\cong P_{|V_H|}$.
\end{lemma}

\begin{proof} Since $G_2$ is a connected graph, $H$ has to be a
connected graph.

For $H\cong P_{|V_H|}$, let $H=P_{u_1u_2\cdots u_{|V_H|}}$. For
$1\leq j\leq |V_H|-1$, since $G_2$ is a 2-connected graph, $G_2$ has
to cross the common bound of $R_j$ and $R_{j+1}$ at least twice.
Hence $\nu_D(E_1,E_2)\geq 2(|V_H|-1)$.

For $H\not\cong P_{|V_H|}$, $|V_H|\geq 3$. For each edge $u_ju_k\in
E_H$, by the definition of $E_H$, $G_2$ crosses the common bound of
$R_j$ and $R_k$ at least once. For $f_n>0$, let $u_j\in V_H$ be an
arbitrary vertex with $d_H(u_j)=1$, $N(u_j)=\{u_k\}$. Since
$|V_H|\geq 3$ and $G_2$ is a connected graph, $d_H(u_k)>1$. Since
$G_2$ is a 2-connected graph, $G_2$ crosses the common bound of
$R_j$ and $R_k$ at least twice. Hence $\nu_D(E_1,E_2)\geq
|E_H|+f_n$.
\end{proof}

\section{Crossing number of $P(10,3)$}
\indent \indent In the Figure \ref{fig:(cr(P103)=6)}, we give a
drawing of $P(10,3)$ with 6 crossings. Hence, we have the following
lemma.
\begin{lemma}\label{lemma: cr(P(103)) <= 6}
$cr(P(10,3))\leq 6$.
\end{lemma}
\begin{figure}[ht]
\centering\includegraphics[scale=1.0]{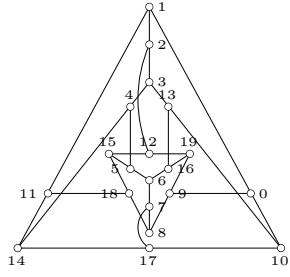}
\caption{\small{A drawing of $P(10,3)$ with 6
crossings}}\label{fig:(cr(P103)=6)}
\end{figure}

In the rest of this section, we shall prove that the value of
$cr(P(10,3))$ is not less than 6.

\begin{lemma}\label{lemma: cr(P(103)) >= 6}
Let $D$ be an arbitrary drawing of $P(10,3)$, then $\nu(D)\geq 6$.
\end{lemma}

\begin{proof}
By contradiction. Suppose $\nu(D)\leq 5$. Then by Lemma \ref{lemma:
cr(E1,E2)=0}, $\nu_D(E_1,E_2)>0$. It follows $|V_H|\geq 2$. By Lemma
\ref{lemma: cr(E_1,E_2)>=E_H}, $\nu_D(E_1,E_2)\geq 2$. Hence
$\nu_D(E_1)+\nu_D(E_2)\leq 3$. Without loss of generality, we assume
$\nu_D(E_1)\leq \nu_D(E_2)$. Then $\nu_D(E_1)=0$ or $\nu_D(E_1)=1$.

\noindent{Case 1. } $\nu_D(E_1)=0$. It follows $2\leq |V_H|\leq 4$,
because $D_1$ has 4 regions in this case. Notice that $H$ is a
connected graph, we can enumerate all the possibilities of $H$. If
$|V_H|=2$, we have $H\in\{ P_2\}$. If $|V_H|=3$, we have $H\in\{
P_3,K_3\}$. And if $|V_H|=4$, we have $H\in\{
P_4,K_4,K_{1,3},K_{1,3}+e,C_4,K_4-e\}$. Since $\nu(D)\leq 5$, by
Lemmas \ref{lemma: cr(E1,E2)=0} and \ref{lemma: cr(E_1,E_2)>=E_H},
$H$ is not isomorphic to any graph of $\{K_{1,3},P_4,K_4\}$, and
there are six subcases depending on $H$.

\noindent{Case 1.1.} $H\cong P_2$. By symmetry, we may assume $G_2$
lies in region $R_1\cup R_2$ and crosses $P_{v_{10}v_{17}v_{14}}$.
By Lemma \ref{lemma: cr(E_1,E_2)>=E_H}, $\nu_D(E_1,E_2)\geq 2$. Let
$R$ be the region lying in outside of $D_2$. By Lemma \ref{lemma:
cr(G12)
>= 3}, $\nu_D(E_2)+\nu_D(E_2,E_{12})\geq
|V_{in}(R;G_{1})|-3\geq5-3=2$. $v_2v_{12}$ crosses one edge of
$C_{v_{10}v_{0}v_{1}v_{11}v_{14}v_{4}v_{3}v_{13}v_{10}}$. Since
$\nu(D)\leq 5$, vertices $v_9$ and $v_{18}$ have to lie in $R_1$,
while vertices $v_5$ and $v_{16}$ have to lie in $R_2$.  By Lemma
\ref{lemma: two 2-degree >=3}, $\nu_D(E_1,E_2)\geq3$. Hence, we have
$\nu(D)\geq 3+2+1=6$, a contradiction (see Figure \ref{fig:(in
2,3)}(1)).

\begin{figure}[h]
\centering\includegraphics[scale=1.0]{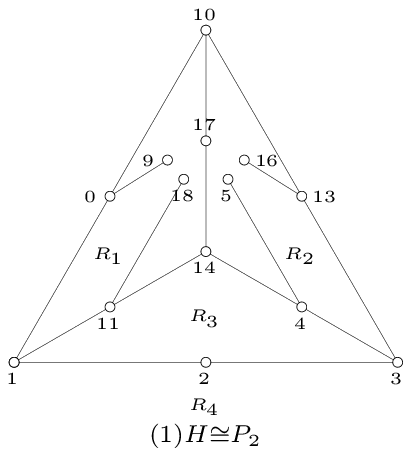}
\centering\includegraphics[scale=1.0]{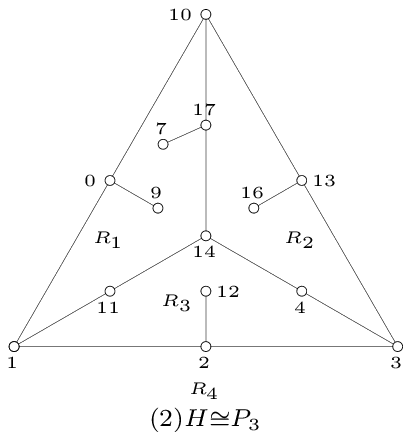}
\centering\includegraphics[scale=1.0]{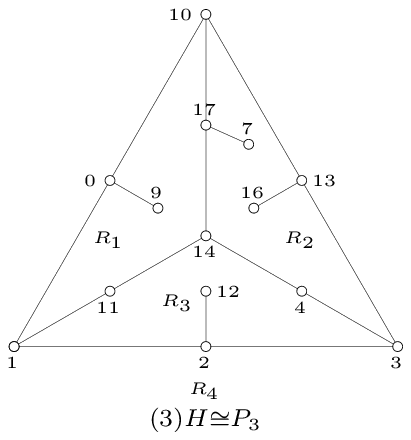} \caption{\small{
A drawing with $H\cong P_2$} and two drawings with $H\cong
P_3$}\label{fig:(in 2,3)}
\end{figure}

\noindent{Case 1.2.} $H\cong P_{3}$. By symmetry, we may assume
$G_2$ lies in region $R_1\cup R_2\cup R_3$, $G_2$ crosses
$P_{v_{14}v_{11}v_{1}}$ and $P_{v_{14}v_{4}v_{3}}$.

By Lemma \ref{lemma: cr(E_1,E_2)>=E_H}, $\nu_D(E_1,E_2)\geq 4$. Let
$R$ be the region lying in outside of $D_2$. By Lemma \ref{lemma:
cr(G12) >= 3}, $\nu_D(E_2)+\nu_D(E_2,E_{12})\geq
|V_{in}(R;G_{1})|-3\geq4-3=1$. Since $\nu(D)\leq 5$, vertex $v_{9}$
has to lie in $R_1$, vertex $v_{12}$ has to lie in $R_3$, vertex
$v_{16}$ has to lie in $R_2$, and $v_{7}$ has to lie in $R_1$
$(R_2)$. By Lemma \ref{lemma: two 2-degree
>=3}, the edges of $E_2$ cross the bound of $R_1$ $(R_2)$
at least three times and cross the edges of $P_{v_{14}v_{4}v_{3}}$
$(P_{v_{14}v_{11}v_{1}})$ at least twice. Hence, we have $\nu(D)\geq
3+2+1=6$, a contradiction (see Figure \ref{fig:(in 2,3)}(2), (3)).

\noindent{Case 1.3.} $H\cong K_{3}$. By symmetry, we may assume
$G_2$ lies in region $R_1\cup R_2\cup R_3$, and $G_2$ crosses
$P_{v_{14}v_{11}v_{1}}$, $P_{v_{14}v_{4}v_{3}}$ and
$P_{v_{14}v_{17}v_{10}}$. By Lemma \ref{lemma: cr(E_1,E_2)>=E_H},
$\nu_D(E_1,E_2)\geq 3$.

Let $t_i=|\{v:v\in V^{3}_2$ and $v$ lies in $R_i\}|$. By symmetry,
we may assume $t_1\leq t_2\leq t_3$. There are four subcases
depending on $(t_1,t_2,t_3)$.

\noindent{Case 1.3.1.} $(t_1,t_2,t_3)=(0,0,4)$. By Lemma \ref{lemma:
cr(G12) >= 3}, $\nu_D(E_1)+\nu_D(E_1,E_{12})\geq
|V_{in}(R_3;G_{2})|-3\geq5-3=2$. Since $\nu(D)\leq 5$, $G_2$ does
not cross itself and all edges of
$\{v_{17}v_{7},v_{11}v_{18},v_{4}v_{5}\}$ are clean. It follows
$v_{2}v_{12}$ has to be crossed. Hence, we have $\nu(D)\geq
3+2+1=6$, a contradiction (see Figure \ref{fig:(0,0,4)}(1)).

\begin{figure}[ht]
\centering\includegraphics[scale=1.0]{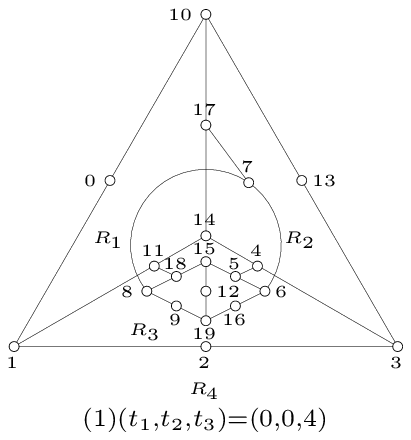}
\centering\includegraphics[scale=1.0]{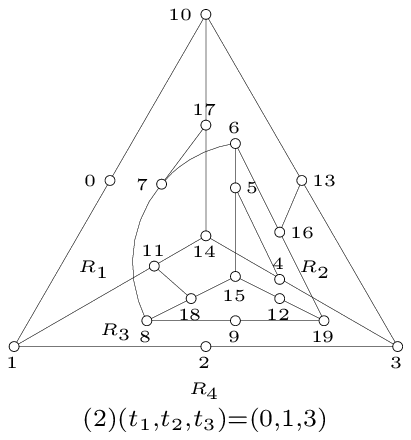}
\centering\includegraphics[scale=1.0]{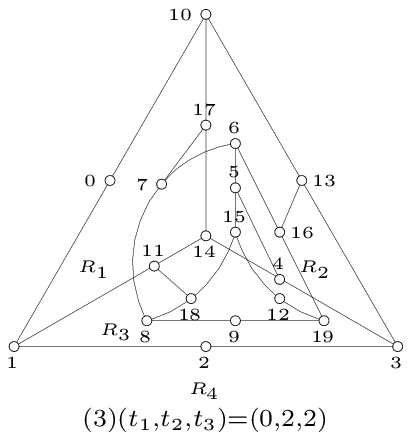} \caption{\small{
Three drawings with $(t_1,t_2,t_3)\in
\{(0,0,4),(0,1,3),(0,2,2)\}$}}\label{fig:(0,0,4)}
\end{figure}

\noindent{Case 1.3.2.} $(t_1,t_2,t_3)=(0,1,3)$. Then the edges of
$E_2$ have to cross the edges of $E_1$ at least four times, i.e.
$\nu_D(E_1,E_2)\geq4$. Let $v\in V^3_2$ be the vertex lying in
$R_2$, then all vertices of $V^{2}_2-N(v)$ lie in $R_3$. By Lemma
\ref{lemma: boundary to 3-degree cr>=1}, $\nu_D(E_1,E_{12})\geq1$.
Since $\nu(D)\leq 5$, $G_2$ does not cross itself and all edges of
$\{v_{13}v_{16},v_{17}v_{7},v_{4}v_{5},v_{11}v_{18}\}$ are clean. It
follows $v_{2}v_{12}$ has to be crossed. Hence, we have $\nu(D)\geq
4+1+1=6$, a contradiction (see Figure \ref{fig:(0,0,4)}(2)).

\noindent{Case 1.3.3.} $(t_1,t_2,t_3)=(0,2,2)$. Then the edges of
$E_2$ have to cross the edges of $E_1$ at least five times, i.e.
$\nu_D(E_1,E_2)\geq5$. Since $\nu(D)\leq 5$, $G_2$ does not cross
itself and all edges of
$\{v_{13}v_{16},v_{17}v_{7},v_{4}v_{5},v_{11}v_{18}\}$ are clean. It
follows $v_{2}v_{12}$ has to be crossed. Hence, we have $\nu(D)\geq
5+1=6$, a contradiction (see Figure \ref{fig:(0,0,4)}(3)).

\noindent{Case 1.3.4.} $(t_1,t_2,t_3)=(1,1,2)$. Then the edges of
$E_2$ have to cross the edges of $E_1$ at least five times, i.e.
$\nu_D(E_1,E_2)\geq5$. Since $\nu(D)\leq 5$, $G_2$ does not cross
itself and all edges of
$\{v_{13}v_{16},v_{17}v_{7},v_{4}v_{5},v_{11}v_{18},v_{0}v_{9}\}$
are clean. It follows $v_{2}v_{12}$ has to be crossed. Hence, we
have $\nu(D)\geq 5+1=6$, a contradiction (see Figure
\ref{fig:(1,1,2)}(1)).

\begin{figure}[h]
\centering\includegraphics[scale=1.0]{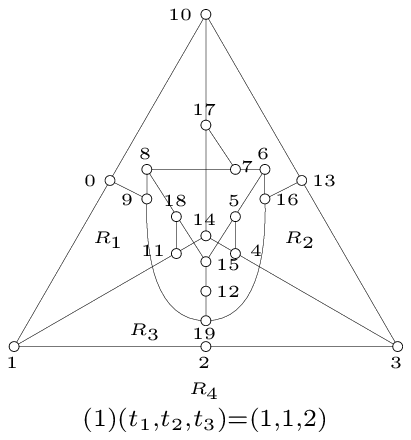}
\centering\includegraphics[scale=1.0]{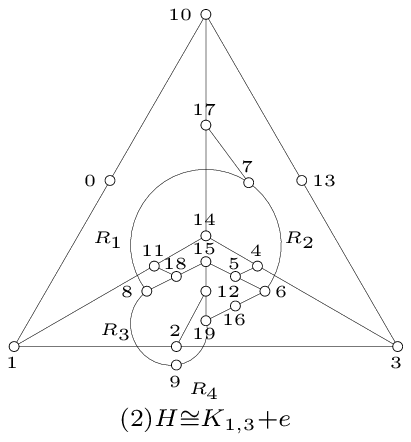}
\centering\includegraphics[scale=1.0]{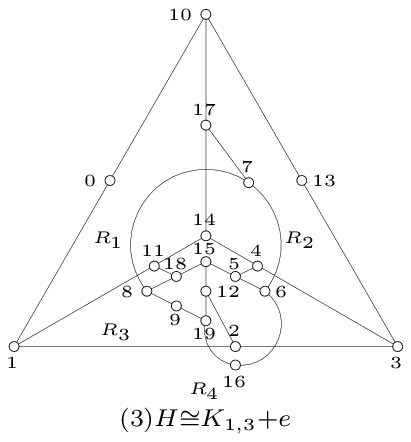}
\caption{\small{A drawing with $H\cong K_3$ and two drawings with
$H\cong K_{1,3}+e$}}\label{fig:(1,1,2)}
\end{figure}

\noindent{Case 1.4.} $H\cong K_{1,3}+e$. By symmetry, we may assume
$G_2$ crosses $P_{v_{14}v_{11}v_{1}}$,
$P_{v_{14}v_{4}v_{3}}$,$P_{v_{14}v_{17}v_{10}}$ and
$P_{v_{1}v_{2}v_{3}}$. By Lemma \ref{lemma: cr(E_1,E_2)>=E_H},
$\nu_D(E_1,E_2)\geq 5$. Since $\nu(D)\leq 5$, $G_2$ does not cross
itself and all edges of
$\{v_{17}v_{7},v_{11}v_{18},v_{4}v_{5},v_{2}v_{12}\}$ are clean. It
follows at least one edge of $v_{13}v_{16}$ and $v_{0}v_{9}$ has to
be crossed. Hence, we have $\nu(D)\geq 5+1=6$, a contradiction (see
Figure \ref{fig:(1,1,2)}(2), (3)).

\noindent{Case 1.5.} $H\cong C_{4}$. By symmetry, we may assume
$G_2$ crosses $P_{v_{10}v_{0}v_{1}}$, $P_{v_{14}v_{11}v_{1}}$,
$P_{v_{14}v_{4}v_{3}}$ and $P_{v_{10}v_{13}v_{3}}$. By Lemma
\ref{lemma: cr(E_1,E_2)>=E_H}, $\nu_D(E_1,E_2)\geq 4$.

Let $t_i=|\{v:v\in V^{3}_2$ and $v$ lies in $R_i\}|$. By symmetry,
there are only two possible kinds of drawings for
$\nu_D(E_1,E_2)\leq 5$.

\noindent{Case 1.5.1.} $(t_1,t_2,t_3,t_4)=(0,0,0,4)$. By Lemma
\ref{lemma: cr(E_1,E_2)>=E_H}, $\nu_D(E_1,E_2)\geq 4$. By Lemma
\ref{lemma: cr(G12) >= 3}, $\nu_D(E_1)+\nu_D(E_1,E_{12})\geq
|V_{in}(R_4;G_{2})|-3\geq5-3=2$. Hence, we have $\nu(D)\geq 4+2=6$,
a contradiction (see Figure \ref{fig:(0,0,0,4)}(1)).

\begin{figure}[ht]
\centering\includegraphics[scale=1.0]{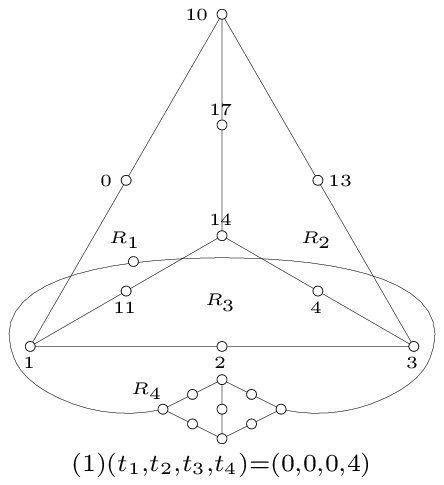}
\centering\includegraphics[scale=1.0]{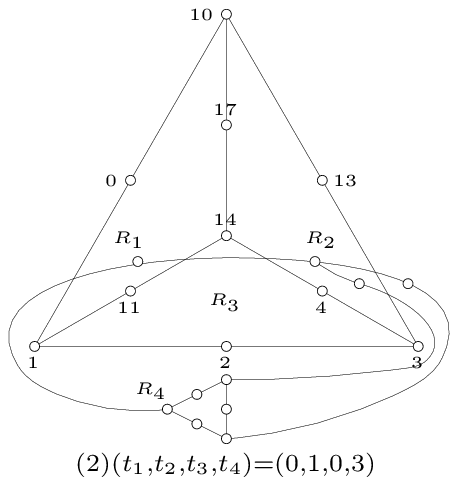}
\centering\includegraphics[scale=1.0]{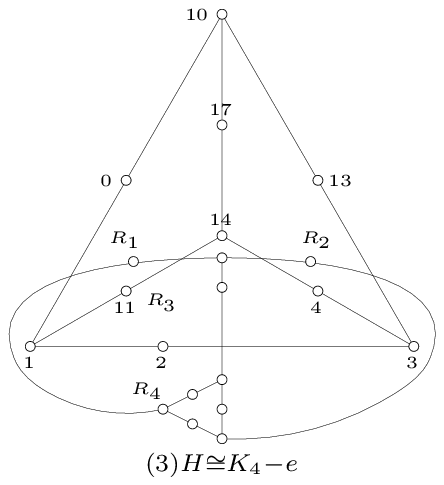}\caption{\small{Two
drawings with $H\cong C_{4}$ and a drawing with $H\cong K_{4}-e$
}}\label{fig:(0,0,0,4)}
\end{figure}

\noindent{Case 1.5.2.} $(t_1,t_2,t_3,t_4)=(0,1,0,3)$. Then $G_2$ has
to cross edges of $E_1$ at least five times, i.e.
$\nu_D(E_1,E_2)\geq5$. Let $v\in V^3_2$ be the vertex lying in
$R_2$, then all vertices of $V^{2}_2-N(v)$ lie in $R_4$. By Lemma
\ref{lemma: boundary to 3-degree cr>=1}, $\nu_D(E_1,E_{12})\geq 1$.
Hence, we have $\nu(D)\geq 5+1=6$, a contradiction (see Figure
\ref{fig:(0,0,0,4)}(2)).

\noindent{Case 1.6.} $H\cong K_{4}-e$. By symmetry, we may assume
$G_2$ crosses $P_{v_{10}v_{0}v_{1}}$, $P_{v_{14}v_{11}v_{1}}$,
$P_{v_{14}v_{4}v_{3}}$, $P_{v_{10}v_{13}v_{3}}$ and
$P_{v_{1}v_{2}v_{3}}$. By Lemma \ref{lemma: cr(E_1,E_2)>=E_H},
$\nu_D(E_1,E_2)\geq 5$. Let $v\in V^3_2$ be the vertex lying in
$R_3$, then all vertices of $V^{2}_2-N(v)$ lie in $R_4$. By Lemma
\ref{lemma: boundary to 3-degree cr>=1}, $\nu_D(E_1,E_{12})\geq 1$.
Hence, we have $\nu(D)\geq 5+1=6$, a contradiction (see Figure
\ref{fig:(0,0,0,4)}(3)).

\noindent{Case 2. } $\nu_D(E_1)=1$. Then $\nu_D(E_1)+\nu_D(E_2)\geq
2$. Since $\nu(D)\leq 5$, $\nu_D(E_1,E_2)\leq 3$. By Lemma
\ref{lemma: cr(E_1,E_2)>=E_H}, $H$ has to be isomorphic to $P_2$ or
$K_3$.

\begin{figure}[ht]
\centering\includegraphics[scale=1.0]{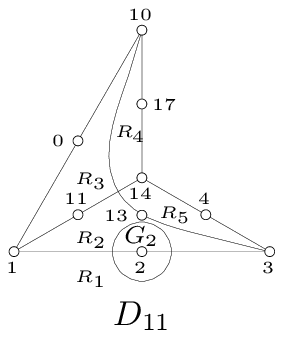}
\centering\includegraphics[scale=1.0]{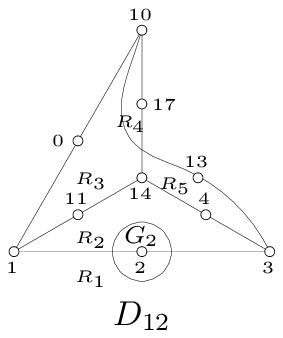}
\centering\includegraphics[scale=1.0]{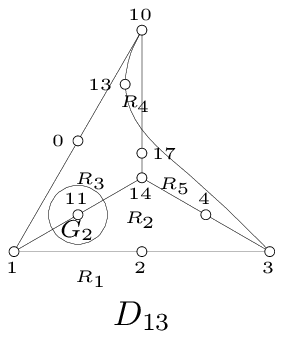}
\centering\includegraphics[scale=1.0]{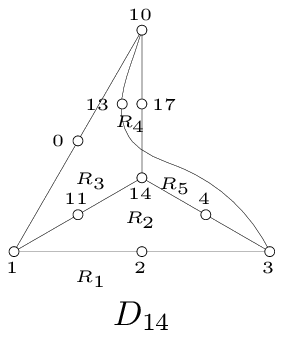}

\caption{\small{Four drawings with
$\nu_D(E_1)=1$}}\label{fig:D(E_1)=1}
\end{figure}

\noindent{Case 2.1.} $H\cong P_2$. By Lemma \ref{lemma:
cr(E_1,E_2)>=E_H}, $\nu_D(E_1,E_2)\geq 2$. By symmetry, there are
only four possible drawings of $G_1$, namely $D_{1i}$ $(i=1,2,3,4)$,
as shown in Figure \ref{fig:D(E_1)=1}. Suppose $G_2$ lies in region
$R_j\cup R_k$ and crosses the common bound of $R_j$ and $R_k$. Let
$R$ be the region lying in the outside of $D_2$. By Lemma
\ref{lemma: cr(G12)
>= 3}, $\nu_D(E_2)+\nu_D(E_2,E_{12})\geq
|V_{in}(R;G_{1})|-3\geq5-3=2$. Since $\nu(D)\leq 5$,
$\nu_D(E_1,E_{12})=0$. For drawing $D_{11}$ and $D_{12}$, the only
possible $(R_j,R_k)$ is $(R_1,R_2)$. For drawing $D_{13}$, the only
possible $(R_j,R_k)$ is $(R_2,R_3)$. For drawing $D_{14}$, no
$(R_j,R_k)$ satisfies $\nu_D(E_1,E_{12})=0$.

For drawing $D_{11}$, $v_{18}$ and $v_{16}$ have to lie in $R_2$,
while $v_{5}$, $v_{7}$ and $v_{9}$ have to lie in $R_1$. For drawing
$D_{12}$, $v_{18}$ and $v_{5}$ have to lie in $R_2$, while $v_{16}$,
$v_{7}$ and $v_{9}$ have to lie in $R_1$. For drawing $D_{13}$,
$v_{12}$ and $v_{5}$ have to lie in $R_2$, while $v_{16}$, $v_{7}$
and $v_{9}$ have to lie in $R_3$. By Lemma \ref{lemma: two 2-degree
>=3}, $\nu_D(E_1,E_2)\geq3$. Hence we have $\nu(D)\geq
1+2+3=6$, a contradiction (see Figure \ref{fig:D(E_1)=1}).

\noindent{Case 2.2.} $H\cong K_3$. By Lemma \ref{lemma:
cr(E_1,E_2)>=E_H}, $\nu_D(E_1,E_2)\geq 3$. Suppose $G_2$ lies in
region $R_j\cup R_k\cup R_l$. For $t=j,k,l$, let $s_t=|\{v: v\in
V^{2}_2$ and $v$ lies in region $R_t\}|$. Without loss of
generality, we may assume $s_j\leq s_k\leq s_l$. If $s_l\geq 5$, by
Lemma \ref{lemma: cr(G12) >= 3}, $\nu_D(E_1)+\nu_D(E_1,E_{12})\geq
|V_{in}(R_l;G_{2})|-3\geq5-3=2$. Hence, we have $\nu(D)\geq
3+1+2=6$, a contradiction. If $s_l\leq 4$, by Lemma \ref{lemma: two
2-degree
>=3}, the edges of $E_2$ have to cross the bound of $R_l$ at least
three times. It follows $\nu_D(E_1,E_2)\geq 3+1=4$. Hence we have
$\nu(D)\geq 2+4=6$, a contradiction (see Figure \ref{fig:H=K_3}).

\begin{figure}[ht]
\centering\includegraphics[scale=1.0]{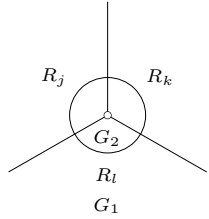}
\caption{\small{$H\cong K_3$}}\label{fig:H=K_3}
\end{figure}

By Cases 1-2, we have $\nu(D)\geq 6$.
\end{proof}

By Lemmas \ref{lemma: cr(P(103)) <= 6} and \ref{lemma: cr(P(103)) >=
6}, we have Theorem \ref{Theorem: cr(P(103))= 6} holds.

{\noindent \bf Acknowledgements}

We are very grateful to the referees for their careful reading with
corrections and useful comments.


\begin{thebibliography}{99}
\bibitem{F86}
Fiorini, S.
\newblock On the crossing number of generalized Petersen
graphs,
\newblock {\it Ann. Discrete Math.} 30, 221-242 (1986).
\bibitem{FG03}
Fiorini, S., Gauci, J.B.
\newblock The crossing number of the generalized
Petersen graph $P[3k,k]$,
\newblock {\it Math. Bohem.} 128, no. 4, 337-347 (2003).
\bibitem{LY09}
Lin, X., Yang, Y., Zheng, W., Shi, L., Lu, W.
\newblock The crossing number of generalized Petersen
graphs with small order,
\newblock {\it Discrete Applied Math.} 157, 1016-1023 (2009).
\bibitem{MR92}
McQuillan, D., Richter, R.B.
\newblock On the crossing number of certain generalized Petersen
graphs,
\newblock {\it Ann. Discrete Math.} 104, 311-320 (1992).
\bibitem{RS02}
Richter, R.B., Salazar, G.
\newblock The crossing number of P(N,3),
\newblock {\it Graphs Combin.} 18, 381-394 (2002).
\bibitem{S97}
Sara\u{z}in, M.L.
\newblock The crossing number of the generalized Petersen
graph P(10,4) is four,
\newblock {\it Math. Slovaca} 47, 189-192 (1997).

\end{thebibliography}
\end{document}